\documentclass{IEEEtran}

\usepackage{amsmath, amsthm, amsfonts, amssymb, amsbsy}
\usepackage{setspace}
\usepackage{graphicx}

\usepackage{pstricks,arydshln}
\usepackage{algorithm}
\usepackage{multirow}
\usepackage{algorithmic}
\usepackage{algorithm}

\newtheorem{theorem}{Theorem}

\newtheorem{example}{Example}

\newcommand{\twotriangle}{\hfill $\bigtriangleup \bigtriangleup$  }
\newcommand{\eax}{\twotriangle  \end{example}}
\newcommand\bim{\begin{itemize}}
\newcommand\eim{\end{itemize}}

\newcommand\bH{{\bf H}}

\newcommand\bA{{\bf A}}

\newcommand\bI{{\bf I}}
\newcommand\bG{{\bf G}}

\newcommand\bD{{\bf D}}
\newcommand\bB{{\bf B}}
\newcommand\bP{{\bf P}}

\newcommand\bV{{\bf V}}
\newcommand\bW{{\bf W}}
\newcommand\ba{{\bf a}}

\newcommand\bc{{\bf c}}
\newcommand\bd{{\bf d}}

\newcommand\bm{{\bf m}}

\newcommand\bw{{\bf w}}

\newcommand\bz{{\bf z}}

\title{
A Low-Complexity Encoding of Quasi-Cyclic Codes Based on Galois Fourier Transform
}

\begin{document}

\author{
Qin Huang$^1$,  Li Tang$^1$, Zulin Wang$^1$, Zixiang Xiong$^2$, \emph{Fellow, IEEE}, and Shanbao He$^3$\\
$^1$School of Electronic and Information Engineering,
Beihang University, Beijing, China, 100191\\
$^2$Dept of ECE, Texas A\&M University, College Station, TX, USA, 77843\\
$^3$China Academy of Space Technology, Beijing, China, 100048\\
(email:qhuang.smash@gmail.com; neathe@163.com; wzulin\_201@163.com; \\ zx@ece.tamu.edu; heshanbao@cast.cn)\\
\thanks{ This work was supported by National Natural Science Foundation of China (61201156).
}
}

\maketitle

\begin{abstract}
The encoding complexity of a general ($en$,$ek$) quasi-cyclic code is $O(e^2(n-k)k)$.
This paper presents a novel low-complexity encoding algorithm for quasi-cyclic (QC) codes based on matrix transformation.
First, a message vector is encoded into a transformed codeword in the transform domain.
Then, the transmitted codeword is obtained from the transformed codeword by the inverse Galois Fourier transform.
For binary QC codes, a simple and fast mapping is required to post-process the transformed codeword such that the transmitted codeword is binary as well.
The complexity of our proposed encoding algorithm is $O(e(n-k)k)$ symbol operations for non-binary codes and $O(e(n-k)k\log_2 e)$ bit operations for binary codes. These complexities are much lower than their traditional counterpart $O(e^2(n-k)k)$.
For example, our complexity of encoding a 64-ary (4095,2160) QC code is only $1.59\%$ of that of traditional encoding, and
our complexities of encoding the binary (4095, 2160) and (8176, 7154) QC codes are respectively
$9.52\%$ and $1.77\%$ of those of traditional encoding.
We also study the application of our low-complexity encoding algorithm to one of the most important subclasses of QC codes, namely QC-LDPC codes, especially when their parity-check matrices are rank deficient.
\end{abstract}


\section{Introduction}

 \emph{Quasi-cyclic (QC)} codes \cite{Townsend_67} are an important class of linear error-correcting codes in both coding theory and their applications.
These codes can asymptotically approach the Varshamov-Gilbert bound  \cite{Kasami74}.
Moreover, their partial cyclic structure simplifies their encoding and decoding implementations by using simple shift registers and logic circuits \cite{LC04}.
In recent years, research on QC codes has focused on one of their subclasses, known as QC low-density parity-check (LDPC) codes \cite{Bon08}-\cite{LZTCLG07}, which have been shown to perform as well as other types of LDPC codes  in most applications.  QC-LDPC codes have advantages over other types of LDPC codes in hardware implementation of encoding \cite{LCZLF06} and decoding \cite{TLZLG06},  \cite{LZTCLG07}.
Thus, most LDPC codes adopted as standard codes for various next-generation communication and storage systems are QC.

One important development of QC codes is the introduction of matrix transformation via the \emph{Galois Fourier transform}  \cite{K07, HLG_10}.
In the Fourier transform domain, an array of circulants is specified by a diagonal matrix over a finite field. It has been shown to be amenable to analysis and construction of some QC-LDPC codes \cite{HLG_12}. Furthermore, the authors of \cite{HLW_12} expanded the analyses of ranks and row-redundancies into most QC-LDPC codes.

QC codes generally are encoded by multiplying a message vector $\bm$ of length $ek$ with an $ek\times en$  generator matrix $\bG$, where $\bG$ usually is systematic \cite{LC04} with $\bG=[\bI \;\vdots\; \bP]$ and $\bI$ being the identity matrix $-$
in the rest of the paper, $\bG$ is systematic if not stated otherwise.
There are two issues with the implementation of QC code systems.
First, the generator matrix $\bG$ is usually not sparse, so it requires a large number, i.e., $e(n-k)k$, of memory units to store $\bP$.
Second, although encoding of QC codes can be partially parallelized so that the computation units are reduced by a factor of $e$, the total number of symbol operations is still $e^2(n-k)k$, which is the same as that for general linear codes.

In this paper, we propose to encode QC codes in the Fourier domain rather than using direct multiplications in the symbol domain.
We are motivated by the fact that the row space of the transformed generator matrix is the null space of the transformed parity-check matrix.
Moreover, the $ek\times en$ transformed generator matrix is an $e\times e$ diagonal array of $k\times n $ matrices.
Consequently, \emph{encoding in the transform domain (ETD)} is achieved by $e$ times encoding of message vectors of length $k$ with $k\times n$ generator matrices (rather than encoding of an $ek$ message vector with an $ek\times en$ generator matrix).
We call the resulting vector after ETD the \emph{transformed codeword}.
The final transmitted QC codeword is obtained by permutations and $n$ times inverse Galois Fourier transforms from the transformed codeword.
Thus, the computational complexity of ETD is $O(e(n-k)k)$, which is much lower than that of traditional encoding.
The memory consumption of ETD is $O(e(n-k)k)$ symbols, which is the same as that of traditional encoding.
For example, we show that the computational complexity of ETD of a 64-ary (4095,2160) QC code is only $1.59\%$ of that of traditional encoding.

Since binary QC codes are used in many applications, we carry out a detailed study of ETD of binary QC codes.
For any binary QC code, its transformed generator matrix satisfies the conjugacy constraint \cite{B83, HLG_12}, but the transmitted codeword from ETD usually is not binary.
Thus, it costs several bits to transmit a code symbol, resulting in lower code rate.
To make the transmitted codeword binary, we devise a simple and fast mapping to post-process the transformed codeword so that the new transformed codeword still satisfies the conjugacy constraint.
The post-processing step consisting of a mapping with bases of subfields is the key in the binary case.
The other steps of ETD of binary QC codes are the same as those for the non-binary case.
Furthermore, if we take advantages of the conjugacy constraint on the transformed generator matrix, the computational complexity of ETD can be reduced to $O(e(n-k)k\log_2 e)$.
Its memory consumption is $O(e(n-k)k)$, the same as in traditional one.
We show that the computation complexities of ETD of the binary (4095, 2160) and (8176, 7154) QC codes are respectively
$9.52\%$ and $1.77\%$ of those of traditional encoding.

ETD readily applies to QC-LDPC codes with full-rank parity-check matrices.
However, there are many QC-LDPC codes, especially algebraic LDPC codes \cite{TLZLG06,LZTCLG07,HLG_10}, \cite{Chen_Lan_Djurdjevic_Lin_04}-\cite{Ga62}, whose parity-check matrices are abundant in redundant rows, i.e., rank deficient.
By carefully constructing the transformed generator matrix according to the rank of the diagonal matrices on the transformed parity-check matrix, we show that our proposed ETD algorithm works for such QC-LDPC (or QC) codes as well.

The rest of this paper is organized as follows.
Section II introduces QC codes and the matrix transformation.
In Section III, we present encoding of QC codes in the transform domain and study its computational complexity and memory consumption. Section IV focuses on ETD of binary QC codes and its simplification due to the conjugacy constraint.
The key step in the binary case that involves construction of a post-processing mapping using subfield bases is derived.
Section V is concerned with QC-LDPC codes, especially whose parity-check matrices are rank deficient.
Section VI concludes the paper.

\section{Introduction of Matrix Transformation}

   In this section, we briefly describe QC codes and matrix transformations of their generator matrices and parity-check matrices and refer readers to \cite{HLG_10, HLG_12} for more details.

   Let $\bW  = [w_{ij}]$, $0 \leq  i, j < e$, be an $e\times e$ circulant matrix over GF($q$), i.e., if every row is a  \emph{cyclic-shift} (one place to right) of the row above it, including end-around. Then, we write $\bW  = circ(\bw)$, where $\bw$ is its top row, called the \emph{generator} of $\bW $ \cite{Tan88}.  A binary ($en$,$ek$) QC code of length $en$ is given by a $k\times n$ array $\bG$ of $e\times e$ \emph{circulant} matrices (or simply circulants), named the generator matrix.
    \begin{equation}\label{eq:G}
\begin{array}{lll}
             \bG &=&
\left[\begin{array}{llll}
    \bW_{0,0}  &     \bW_{0,1} &  ... &     \bW_{0,n-1} \\
         \bW_{1,0} &   \bW_{1,1} &  ...&     \bW_{1,n-1}   \\
          \vdots& &\ddots & \vdots\\
                                       \bW_{k-1,0} &    \bW_{k-1,1}&  ...&  \bW_{k-1,n-1}
\end{array}\right],
                          \end{array}
\end{equation}
   or the null space of an $(n-k)\times n$ array $\bH=[\bA_{i,j}]$, parity-check matrix,  of circulants of the same size, where $\bG \cdot \bH^{\sf T} ={\bf 0}$
and $\bA_{i,j}=circ(\ba_{i,j})$ is a circulant of length $e$.

Consider a finite field of $2^r$ elements GF($2^r$), i.e., $q=2^r$. Let $\alpha$  be an element of GF($2^r$) with order $e$ and $\bw$ be an $e$-tuple over GF($2^r$), where $e$ is a factor of $2^r-1$. Here we consider only the case $e=2^r-1$. The other cases are similar.  The  Galois Fourier transform of the $e$-tuple $\bw$ \cite{B83} over GF($2^r$), denoted by ${\cal  F}[\bw]$,  is given by the $e$-tuple $\bd$ over GF($2^r$) whose $t$-th component, $d_t$, for $0 \leq  t < e$, is given by $d_t = w_0 + w_1\alpha ^{-t} + w_2\alpha ^{-2t} + \cdots + w_{e-1}\alpha ^{-(e-1)t} $. The vector $\bw$, which is called the  \emph{inverse Fourier transform} of the vector $\bd$, denoted by $\bw = {\cal {\cal F}}^{-1}[\bd]$, can be reconstructed from 
$w_l = d_0 + d_1\alpha ^{l} + d_2\alpha ^{2l} + \cdots + d_{e-1}\alpha ^{(e-1)l}$, for $ 0 \leq  l < e$.

Consider an $e\times e$ circulant $\bW= circ(\bw)$. Define two $e\times e$ matrices over GF($2^r$) as follows: $\bV = [\alpha ^{-ij}]$ and $\bV ^{-1} = [\alpha ^{ij}]$,  $0 \leq  i, j <e$. Both matrices, $\bV $ and $\bV ^{-1}$, known as  \emph{Vandermonde} matrices \cite{B83} \cite{RO06}, are non-singular. Moreover, $\bV ^{-1}$ is the inverse of $\bV  $ and vice versa. Taking the matrix product $\bV^{-1}\bW\bV$, we obtain the following $e\times e$ diagonal matrix over GF($2^r$):
\begin{equation}
           {\bW}^{\cal {\cal F}} =  \bV^{-1}\bW\bV  = diag(d_0, d_1, ..., d_{e-1}),
\end{equation}
where the diagonal vector $(d_0, d_1, . . . , d_{e-1})$ is the Fourier transform of the  generator  $\bw$ of $\bW $.  The diagonal matrix $\bW ^{\cal {\cal F}} = \bV^{-1}\bW\bV$ is referred to as the Fourier transform of the circulant $\bW $. If $\bw$ is an $e$-tuple over GF(2), the components must satisfy the following constraint  \cite{B83}
\begin{equation}\label{eq:cc_vct}
                                              d_{(2t)_e} = d_t ^2
\end{equation}
for $0 \leq  t < e$, where $(2t)_e$ denotes the nonnegative integer less than $e $ and is congruent to $2t \mbox{ modulo } e$.  This condition is known as the  \emph{conjugacy constraint}, which is the key constraint of the binary case in Section IV.

   Let $k$ and $n $ be two positive integers. Let $\bG  = [\bW_{i,j}]$, $0 \leq  i < k$, $0 \leq  j < n$, be an $k\times n$ array of $e\times e$ circulants $\bW_{i,j}$ over GF($2^r$), where $\bw_{i,j}$ is the generator of the circulant $\bW_{i,j}$. Next, we define $\Omega (k)$ as a $k\times k$ diagonal array of $\bV $'s and $\Omega ^{-1}(n)$ as an $n\times n$ diagonal array of $\bV ^{-1}$'s,
\begin{equation}
\begin{array}{ccc}
                         \Omega (k) &=& diag( \underbrace{\bV  , . . . , \bV}_{k}  ),   \\
            \end{array}         \end{equation}
\begin{equation}
\begin{array}{ccc}
                         \Omega ^{-1}(n) &=& diag(\underbrace{\bV ^{-1},  . . . , \bV ^{-1}}_{n}). \\
              \end{array}
\end{equation}
Then the Fourier transform of $\bG$  is given as $\bG^{\cal {\cal F}} = \Omega^{-1}(k) \bG \Omega(n)=    [\bW_{i,j}^{\cal F}]$,
where $\bW_{i,j}^{\cal F} = \bV^{-1}\bW_{i,j}\bV$ is an $e\times e$ diagonal matrix with diagonal vector $(d_{i,j,0}, d_{i,j,1}, . . . , d_{i,j,e-1})$, which is the Fourier transform of the generator $\bw_{i,j}$ of $\bW_{i,j}$.

Define the following index sequences: for $0 \leq  i, j < e$, $\pi _{row,i} = [i, e+i, ..., (k-1)e+i]$ and $\pi _{col,j} = [j, e+j, ..., (n-1)e+j]$.
Let $\pi_{ row} = [\pi _{row,0}, \pi_{ row,1}, . . . , \pi_{ row,e-1}]$ and $\pi _{col} = [\pi _{col,0}, \pi_{ col,1}, . . . , \pi_{ col, e-1}]$.
Then $\pi_{row}$ gives a permutation of the indices of the rows of $\bG ^{\cal F}$ while $\pi _{col}$ represents a permutation of the columns of $\bG ^{\cal F}$. Their reverse permutations are denoted by $\pi^{-1}_{row}$ and $\pi^{-1}_{col}$, respectively. We define the permutation $\pi$ that performs both row and column permutations. Its reverse permutation is denoted by $\pi^{-1}$.  By the permutation $\pi$, $\bG^{\cal F}$ results in the following $e\times e$ diagonal array of $k \times n$ matrices over GF($2^r$),
\begin{equation}\label{eq:mt}
\begin{array}{ccl}
                 \bG ^{{\cal F},\pi } &=& diag(\bD_0, \bD_1, . . . , \bD_{e-1}).
\end{array}
\end{equation}
   The transformation from $\bG$  to $\bG ^{{\cal F},\pi }$ through $\bG ^{\cal F}$ is reversible. The reverse process is called the inverse matrix transformation, denoted by $\{{\cal F}^{-1}, \pi^{-1}\}$.

 If  the array $\bG$  of circulants and zero matrices (ZM) is over $\mbox{GF(2)}\subseteq \mbox{GF(}2^r\mbox{)}$, the matrices on the main diagonal of the array $\bG ^{{\cal F},\pi}$ satisfy the conjugacy constraint \cite{B83},
\begin{equation}\label{cc1}
                                         \bD_{(2t)_e} = \bD_t^{\circ 2},
\end{equation}
i.e., the entry at location $(i,j)$ of $\bD_{(2t)_e}$ is the square of the entry at location $(i,j)$ of $\bD_t$.  We call the matrix $\bD_{(2t)_e}$  a \emph{conjugate matrix} of $\bD_t$. Following the definition of conjugate matrix, we can group all the matrices on the main diagonal $\bD_i$, $i=0, 1, \ldots, e-1$, into conjugacy classes.  Let $\lambda$ be the number of distinct conjugacy classes and $\Psi_0, \Psi_1, \ldots, \Psi_{\lambda-1}$ represent these classes, with
\begin{equation} \label{eq:c_psi}
\Psi_i =\{\bD_{t_i},\bD_{(2t_i)_e},...,\bD_{(2^{\eta_i-1}t_i)_e}\}= \{\bD_{t_i}, \bD_{t_i}^{\circ 2}, . . . , \bD_{t_i}^{\circ 2^{\eta_i-1}}\},
\end{equation}
 where $\eta_i$ is the number of matrices in the conjugacy class $\Psi_i$, i.e., $\eta_i$ is the smallest nonnegative integer such that $(2^{\eta_i}t_i)_e=t_i$, and $t_i$ is the smallest number in the subscripts of the conjugate matrices in $\Psi_i$. The member matrix $\bD_{t_i}$ is called the representative of the conjugacy class $\Psi_i$.

\section{Encoding of QC Codes in the Transformed Domain}
In this section, we present our proposed low-complexity ETD algorithm of QC codes. First, we derive ETD based on orthogonality \cite{HLG_12} of the transformed parity-check matrix and the transformed generator matrix. Then, we compare the computational complexity and memory consumption of ETD with those of traditional encoding.

Consider an ($en$, $ek$) QC code $\cal C$ over GF($2^r$) defined by the $ek\times en$ generator matrix $\bG$ in (\ref{eq:G}), which consists of circulants of size $e$. Suppose that $\alpha$ is a primitive element in GF($2^r$), i.e., $e=q-1=2^r-1$. By the matrix transformation (\ref{eq:mt}), the generator matrix results in the transformed generator matrix $\bG^{{\cal F}, \pi}$ over GF($2^r$). It can be employed to encode a message vector $\bm$ of $ek$ bits into a transformed codeword $\bc^{{\cal F}, \pi}$ of length $en$ bits.
Since
$$\bG\cdot\bH^{\sf T}= \Omega (k) \cdot  \Omega ^{-1}(k) \cdot \bG \cdot  \Omega(n)  \cdot \Omega^{-1}  (n)  \cdot\bH^{\sf T}=0,$$
where $\bH$ is the parity-check matrix of $\cal C$, then
\begin{equation}
\Omega(k)  \cdot \bG^{\cal F} \cdot     \Omega^{-1}  (n)\cdot\bH^{\sf T} =0.
\end{equation}
Because $\Omega (k) $ is a nonsingular matrix,
\begin{equation} \label{eq:encoding}
\bG^{\cal F} \cdot     \Omega^{-1}  (n)\cdot\bH^{\sf T} =0.
\end{equation}
Furthermore, permuting the rows of $\bG^{\cal F}$ by $\pi_{row}$ results in  $\bG^{\mathcal{F} , \pi_{row}}$, then\\
\begin{equation} \label{eq:row}
\bG^{\mathcal{F} , \pi_{row}} \cdot     \Omega^{-1}  (n)\cdot\bH^{\sf T} =0.
\end{equation}
Suppose that $\bm$ is a vector of length $ek$ over GF($2^r$). We define the permutation $\pi$ of a vector performs only the column permutations. Its reverse permutation is denoted by $\pi^{-1}$. Then the vector $\bc$, defined as $\bc\triangleq \{\bc^{{\cal F}, \pi}\}^{ \pi^{-1} } \cdot \Omega ^{-1} (n)= \bc^{\cal F} \cdot \Omega^{-1} (n)$ with
\begin{equation} \label{eq:t_c}
\bc^{{\cal F},\pi}= \bm \cdot \bG^{{\cal F}, \pi}
\end{equation}
 is named as \emph{transformed codeword}. The vector $\bc$ is a codeword of the QC code $\cal C$ from (\ref{eq:row}) and (\ref{eq:t_c}), because
\begin{equation}
\begin{array}{lll}
\bc  \cdot\bH^{\sf T}&=& \{\bc^{\cal F, \pi}\}^{ \pi^{-1} } \cdot \Omega^{-1}  (n) \cdot \bH^{\sf T}  \\
&=& \{\bm \cdot \bG^{{\cal F}, \pi} \}^{ \pi^{-1} }  \cdot \Omega^{-1}  (n)\cdot\bH^{\sf T}   \\
   &=&   \bm \cdot \bG^{\mathcal{F} , \pi_{row}} \cdot \Omega ^{-1} (n)\cdot\bH^{\sf T}       \\
   &=&  {\bf 0}.
\end{array}
\end{equation}
Moreover, since $\{ \bG^{{\cal F}, \pi} \}^{ \pi^{-1} }  \cdot \Omega ^{-1}(n)$ has full rank, the mapping from $\bm$ to $\bc$ are one-to-one. As a result, the above equations can be viewed as an encoding in the transform domain. ETD is depicted in Fig. 1 and its steps are summarized as follows:

\begin{algorithm}
\caption{The Encoding of QC Codes in the Transform Domain}
\label{alg:framwork}
\begin{algorithmic}
\REQUIRE ~~\\
  The message $\mathbf{m}$ of $ek$ symbols;\\
  The $ek\times en$ transformed generator matrix, $\mathbf{G}^{\mathcal{F},\pi}$;\\

 \ENSURE ~~\\
  The binary transmitted codeword $\mathbf{c}$ of $en$ symbols;\\
\hspace{-0.1in}\textbf{Steps:}     ~\\

\STATE 1) The message $\bm$ is encoded into the transformed codeword $\bc^{\cal F}$ by the transformed generator matrix $\bG^{{\cal F}, \pi}$
$$\bc^{\cal F}= \{\bm \cdot \bG^{{\cal F}, \pi} \}^{ \pi^{-1} }.  $$

\STATE 2) The transmitted codeword $\bc$ is  obtained by the inverse Galois Fourier transform from ${\bc}^{\cal F}$'s,
\begin{equation}
c_i   = \sum\limits^{e-1}_{j=0}    {c}^{\mathcal F}_{i,j}  \alpha^{j}.
\end{equation}
 \end{algorithmic}
\end{algorithm}

\begin{figure}
\centering
\includegraphics[width=3in]{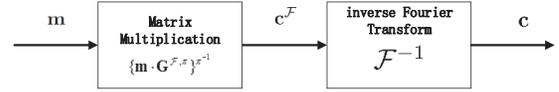}
\caption{The block diagram of the encoding of QC codes in the transformed domain.}
\end{figure}
According the above encoding steps, it is straightforward to give the computational complexity of ETD. The first step of encoding involves $e(n-k)k$ Galois additions and $e(n-k)(k-1)$ Galois multiplications.  Since the computational complexity of the Fourier transform of length $e$ implemented by the Galois fast Fourier transform (GFFT) \cite{KT1990}  is less than $e\log_2 e$ Galois operations, the computational complexity of the second step, which needs $n$ GFFT, is less than $ne\log_2 e$ Galois operations. Because $\log_2 e$ is much smaller than $(n-k)$ or $k$, the computational complexity of  ETD is about $O(ek(n-k))$.

In terms of  the memory consumption of ETD, it is clear that the biggest expenses on memory is for the transformed matrix $\bG^{{\cal F}, \pi}$. Since it is a diagonal matrix, we only have to store $e$ matrices, each of which requires the storage of ($e(n-k)k$) Galois symbols if it is systematic.   In Table I, we compare the computational complexity and the memory consumption of ETD of QC codes with those of traditional encoding in terms of symbols. It is seen that the computational complexity of ETD is about ${\cal R}= e$ times lower than that of traditional encoding.

  \begin{table*}\caption{Complexity of traditional encoding and transformed encoding}\centering
  \begin{tabular}{c||c|c}
  \hline
  &Traditional Encoding &Transformed Encoding\\
  \hline
  \hline
Computational Complexity &$O(e^2 k (n-k))$&  $O(ek((n-k))$\\
\hline
  Memory Consumption                      & $ek(n-k)$  &  $ ek(n-k)$ \\
 \hline
  \end{tabular}
  \end{table*}

\begin{example}
Consider a $64$-ary (4095, 2160) QC code with circulant size 63, i.e., $n=4095$, $k=2160$, $n-k=1953$, $e=63$ and $r=6$. Thus, ${\cal R}=63$. In other words, the computational complexity of ETD is only $1.59\%$ of that of traditional encoding.
\end{example}

\section{Encoding of Binary QC Codes in the Transformed Domain}
In Section III, we derived ETD of QC codes over GF($2^r$). If the QC code is binary, both the message vector $\bm$ and the codeword vector $\bc$ are binary. However, there is no guarantees that the codeword $\bc$ from ETD is binary.
If there exist non-binary symbols in $\bc$, many more bits are needed to transmit it, resulting in code rate reduction.

In this section, we propose a fast and simple mapping using subfield bases on the message vector $\bm$ to make sure that the transmitted codeword $\bc$ is also binary. In addition, to reduce the computational complexity of the first step of ETD,  a post-processing step on the transformed codeword $\bc^{{\cal F},\pi}$, which is equivalent to pre-processing of the message vector $\bm$, will be presented.

Recall the conjugacy constraint (\ref{eq:cc_vct}) for a binary vector. For the transformed codeword $\bc^{\cal F}$ or $\bc^{{\cal F},\pi}$, it means
\begin{equation}\label{eq:conj_vct_1}
c^{{\cal F}}_{ej+(2i)_{e} }=(c^{{\cal F}}_{ej+i})^2 ,
\end{equation}
or \begin{equation}\label{eq:conj_vct_2}
c^{{\cal F},\pi }_{n(2i)_{e} + j }=(c^{{\cal F},\pi}_{ni+j})^2 ,
\end{equation}
 for $0\leq i<e$ and $0\leq j<n$. For the sake of simplicity, we denote the transformed codeword and the codeword by $n$ blocks of length $e$, $\bc^{\cal F}=[\bc^{\cal F}_j]$ and $\bc=[\bc_j]$, respectively, where $c^{\cal F}_{j,i}= c^{\cal F}_{ej+i}$ and $c_{j,i}=c_{ej+i}$. It is clear that $\bc_j$ is the inverse Galois Fourier transform of $\bc^{\cal F}_j$. Similarly, we denote $\bc^{{\cal F},\pi}$ by $e$ blocks of length $n$, $\bc^{{\cal F},\pi}=[\bc^{{\cal F},\pi}_i]$, where $c^{{\cal F},\pi}_i= c^{{\cal F},\pi}_{in+j}$.
 Then (\ref{eq:conj_vct_1}) and (\ref{eq:conj_vct_2}) can be rewritten as
 \begin{equation}\label{eq:conj_block}
c^{\cal F}_{j, (2i)_{e}}=(c^{\cal F}_{j, i})^2,
\end{equation}
and
 \begin{equation}\label{eq:conj_block2}
c^{{\cal F},\pi}_{(2i)_{e},j}=(c^{{\cal F},\pi}_{i, j})^2,
\end{equation}
 respectively.

We now show that if the message vector $\bm$ is pre-processed by bases of subfields, then (\ref{eq:conj_block}) or (\ref{eq:conj_block2}) is satisfied and hence the codeword $\bc$ is binary.
Again for the sake of simplicity, we denote the message vector $\bm$ and the pre-processed message vector $\hat{\bm}$  by $e$ blocks of size $k$, $\bm=[\bm_i]$ and $\hat{\bm} = [\hat{\bm}_i]$, respectively, where $m_{i,j}=m_{ik+j}$ and $\hat{m}_{i,j}=\hat{m}_{ik+j}$, $i=0,1,\ldots,e-1$ and $j=0,1,\ldots, k-1$.

Recall the definition of the $i$-th conjugacy class $\Psi_i=\{\bD_{t_i}, \bD_{t_i}^{\circ 2}, . . . , \bD_{t_i}^{\circ 2^{\eta_i-1}} \}$ in (\ref{eq:c_psi}) with size $\eta_i$. Suppose that $\alpha$ is a primitive element, then a basis $\beta_{i,0}, \beta_{i,1}, \ldots, \beta_{i, \eta_i-1}$ spans the subfield GF($2^{\eta_i}$) of GF($2^r$) whose element's $2^{\eta_i}$ power all equals to itself. Thus, if $\gamma$ is an element of GF($2^{\eta_i}$), then  $\gamma^{2^{\eta_i}} =\gamma$ and $\gamma= \sum\limits^{\eta_i-1}_{l=0} u_l\beta_{i,l}$, where $u_l$ is in the ground field GF(2). If $\eta_i=1$, then the subfield only has two elements, 0 and 1. If $\eta_i=r$, then the subfield is the field GF($2^r$).

Using the bases $\beta_{i,l}$'s, we map message $\bm$ to its pre-processed version $\hat{\bm}$ via
\begin{equation}\label{eq:m}
\hat{m}_{(2^\mu t_i)_e, j}  =  \left(\sum\limits^{\eta_i-1}_{l=0 }  \beta_{i,l} m_{(2^l t_i)_e, j}   \right)^{2^ \mu}.
\end{equation}
First, the mapping from $\bm$ to $\hat{\bm}$ is one-to-one since $\beta_{i,l}$'s are linearly independent over GF(2).
Second,  $\hat{m}_{(2^\mu t_i)_e, j} = \hat{m}_{(t_i)_e, j}^{2^\mu}$.
A detailed proof of this fact is given the appendix.
The following theorem shows that the transformed codeword
\begin{equation}
\hat{\bc}^{\cal F}= \{\hat{\bm} \cdot \bG^{{\cal F}, \pi} \}^{ \pi^{-1} }
\end{equation}
 encoded from the pre-processed message $\hat{\bm}$ satisfies the conjugacy constraint.

\begin{theorem}
The transformed codeword  $\hat{\bc}^{\cal F}= \{\hat{\bm} \cdot \bG^{{\cal F}, \pi} \}^{ \pi^{-1} } $ encoded from the pre-processed message $\hat{\bm}$ satisfies the conjugacy constraint $\hat{c}^{\cal F}_{j, (2i)_{e}}=(\hat{c}^{\cal F}_{j, i})^2$ in (\ref{eq:conj_vct_1}) or $\hat{c}^{{\cal F},\pi }_{(2i)_e,j }=(\hat{c}^{{\cal F},\pi}_{i,j})^2$ in (\ref{eq:conj_vct_2}).
\end{theorem}
\begin{proof}
From the definition of $\hat{\bc}^{{\cal F}, \pi}$, its $(ni+j)$-th symbol for $i=0,1,\ldots, e-1$ and $0\leq j<n$ is
\begin{eqnarray}
\hat{c}^{{\cal F}, \pi}_{i, j} &=& \sum\limits^{k-1}_{s=0}\hat{m}_{i, s}  D_{ i, s, j }.
\end{eqnarray}
From the definition of $\hat{\bm}$ (\ref{eq:m}) and the conjugacy constraint on $\bD_i$ in (\ref{cc1}),  its $(n(2i)_e+j)$-th symbol is
\begin{equation}
\begin{array}{lll}
\hat{c}^{{\cal F},\pi}_{t'}=\hat{c}^{{\cal F}, \pi}_{(2i)_e, j} &=& \sum\limits^{k-1}_{s=0}\hat{m}_{(2i)_e, s}  D_{ (2i)_e, s, j },\\
     &=& \sum\limits^{k-1}_{s=0} \hat{m}_{i, s}^2  D_{i, s, j}^2,\\
     &=& \left(\sum\limits^{k-1}_{s=0}  \hat{m}_{i, s} D_{i, s, j}\right)^2= (\hat c^{{\cal F},\pi}_{t})^2.
\end{array}
\end{equation}
\end{proof}
Direct computation of (20) involves $k$ Galois multiplications and $k-1$ Galois additions over GF($2^r$), since both $\hat{\bm}$ and $\bD_i$ are non-binary symbols over GF($2^r$).  Thus, we rewrite (21) as
\begin{equation}\label{eq:c}
\begin{array}{lll}
\hat{c}^{{\cal F},\pi}_{{(2^{\mu}t_i)}_e,j}&=&\sum\limits^{k-1}_{s=0} \hat{m}_{{(2^{\mu} t_i)}_e, s} D_{{(2^\mu t_i)}_e, s, j},\\
&=& \sum\limits^{k-1}_{s=0}   \left(\sum\limits^{\eta_i-1}_{l=0 }  \beta_{i,l} m_{{(2^l t_i)}_e, s}   \right)^{2^\mu} D_{t_i, s, j}^{2^\mu},\\
&=&   \sum\limits^{\eta_i-1}_{l=0 }  \beta_{i,l}^{2^\mu}   \left(\sum\limits^{k-1}_{s=0} m_{(2^l t_i)_e, s}^{2^\mu}    D_{t_i, s, j}^{2^\mu}\right),\\
&=&   \left(\sum\limits^{\eta_i-1}_{l=0 }  \beta_{i,l}  c^{{\cal F},\pi}_{(2^l t_i)_e,j}\right)^{2^\mu}.
\end{array}
\end{equation}
Equation (\ref{eq:c}) shows that the mapping on message $\bm$ (\ref{eq:m}) and the mapping on ${\bc}^{{\cal F},\pi}$  (\ref{eq:c}) result in the same $\hat{\bc}^{{\cal F},\pi}$, which satisfies the conjugacy constraint.

However, the mapping on ${\bc}^{{\cal F},\pi}$ (\ref{eq:c}) involves much less Galois multiplications than the one on message $\bm$ (\ref{eq:m}).
According to (\ref{eq:c}), computation of $\hat{c}^{{\cal F}, \pi}_{{t_i},j}$ is carried out in two steps.
In the first step, the summation of $c^{{\cal F},\pi}_{(2^l t_i)_e,j}=\sum\limits^{k-1}_{s=0} m_{(2^l t_i)_e, s}    D_{{(2^\mu t_i)}_e, s, j}$ is calculated, which only involves $k-1$ additions, since $\bm$ is a binary vector; in the second step, $\hat{c}^{{\cal F},\pi}_{ t_i,j}$ is calculated, which involves $\eta_i$ multiplications and $\eta_i-1$ additions.
The other codeword symbol $\hat{c}^{{\cal F},\pi}_{(2^\mu t_i)_e, j}$ in the conjugacy class of $j$ can be simply calculated from $(\hat{c}^{{\cal F},\pi}_{t_i,j})^{2^\mu}$.

Before we study the computational complexity of our proposed ETD algorithm, we summarize its encoding steps.
For the sake of implementational simplicity, we assume that the permutation $ \pi^{-1} $ operation is included
in the initialization stage.
\begin{algorithm}
\caption{The Encoding of Binary QC Codes in the Transform Domain}
\label{alg:framwork_bin}
\begin{algorithmic}
\REQUIRE ~~\\
  The message $\mathbf{m}$ of $ek$ bits;\\
  The $ek\times en$ transformed generator matrix $\mathbf{G}^{\mathcal{F},\pi}$;\\

 \ENSURE ~~\\
  The binary transmitted codeword $\mathbf{c}$ of $en$ bits;\\
\hspace{-0.1in}\textbf{Steps:}     ~\\

\STATE 1) The message $\mathbf{m}$  is encoded into the transformed codeword $\mathbf{c}^{\mathcal{F}}$ by the transformed generator matrix $\mathbf{G}^{\mathcal F, \pi}$
  \[
  \mathbf{c}^{\mathcal F}=\{ \mathbf{m}\cdot {\mathbf G}^{\mathcal{F},\pi}\}^{\pi^{-1}}
  \]
\STATE 2) The transformed codeword $\mathbf{c}^{\mathcal F}$ is mapped into the conjugacy constraint satisfied codeword $\hat {\mathbf{c}}^{\mathcal F}$,
  \[
  \hat{c}_{j,(t_i2^\mu)_e}^{\mathcal{F}}=\bigg(\sum_{l=0}^{\eta_i-1}{\beta_{i,l}c^{\cal F}_{j,(t_i2^l)_e}}\bigg )^{2^\mu}.
  \]
\STATE 3) The binary transmitted codeword $\mathbf{c}$ is obtained by the inverse Galois Fourier transform from $\hat {\mathbf{c}}^{\mathcal F}$'s,
  \[
  c_i=\sum_{j=0}^{e-1}{\hat c_{i,j}^{\cal F}\alpha^{j}}.
  \]
 \end{algorithmic}
\end{algorithm}

\begin{figure*}
\centering
\includegraphics[width=5.5in]{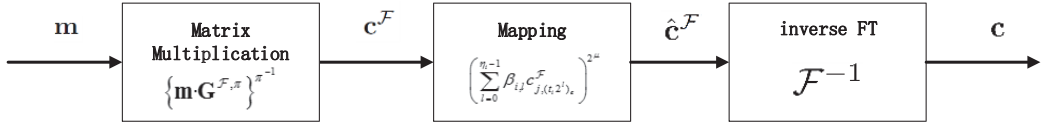}
\caption{The block diagram of the encoding of QC codes in the transformed domain.}
\end{figure*}

We present the computational complexity of the proposed ETD algorithm step by step. Since $\bm$ is binary, the first step of encoding involves only $ek(n-k)$ Galois additions.  In the second step, each  $\hat c^{\cal F}_{j,t_i}$ needs $\eta_i$ Galois multiplications and $\eta_i-1$ Galois additions; the computational complexity of this step is thus about $ ne$ Galois multiplications and $ ne$ Galois additions. Since the GFFT over GF($2^r$) is about $e\log_2 e$ Galois operations, so the complexity of the third step, which requires $n$ times GFFT, is $O(ne\log_2 e)$.

In terms of the memory consumption, it is clear that most memory is spent on storing the transform matrix $\bG^{{\cal F}, \pi}$. Since it is a diagonal matrix, we only need to store $e$ matrices, each of which requires to store ($k(n-k)$) Galois symbols if it is systematic. Moreover, considering the conjugacy constraint (\ref{cc1}), we only need to store $\lambda$  representative matrices. Each symbol in the $i$-th  representative matrix cost $\eta_i$ bits to store.
Thus, the overall memory consumption of ETD is $\sum\limits^{\lambda}_{i=0} \eta_ik(n-k)=ek(n-k)$, which is the same as
that of traditional encoding.

 Considering that each Galois addition costs $r$ bit operations, each Galois multiplication costs $r^2$ bit operations, and each Galois symbol costs $r$ bits memory, we compare in Table II the computational complexity of ETD of binary QC codes with that of traditional encoding in terms of bits.
  \begin{table*}\caption{Complexity of traditional encoding and transformed encoding}\centering
  \begin{tabular}{c||c|c|c}
  \hline
  &Traditional Encoding & \multicolumn{2}{c}{Transformed Encoding}\\
  \hline
  \hline
  \multirow{3}*{Computational Complexity} &\multirow{3}*{$O(e^2k(n-k))$}& Step 1) &$ek(n-k)r$\\
  \cline{3-4}
                                          &             & Step 2) &$ne(r^2+r)$\\
   \cline{3-4}
                                          &             & Step 3) &$n r^2e\log_2 e$  \\
    \cline{3-4}
                                          &             & Overall &  $O(e(n-k)k\log_2 e)$  \\
   \hline
    \hline
  Memory Consumption                      & $ek(n-k)$  & \multicolumn{2}{c}{ $ek(n-k)$ }\\
 \hline
  \end{tabular}
  \end{table*}

  Since $r \approx \log_2 e$, Step 1) costs most computation and dominates the computational complexity, which is about $O(e(n-k)k\log_2 e)$. Thus the complexity of  the transformed encoding is about ${\cal R}= \frac{(n-k)ke^2}{ (n-k)ke\log_2 e} =\frac{e}{\log_2 e}$ times lower than the complexity of the traditional encoding.

\begin{example}
Consider a (4095, 2160) QC-LDPC with circulant size 63, i.e., $n=4095$, $k=2160$, $n-k=1953$, $e=63$ and $r=6$, we have, ${\cal R}=10.05$. In other words, the computational complexity of ETD is only $9.52\%$ of that of traditional encoding.
\end{example}

\begin{example}
Consider a (8176, 7154) QC-LDPC with circulant size 511, i.e., $n=8176$, $k=7154$, $n-k=1022$, $e=511$ and $r=9$, we have, ${\cal R}=56.78$. In other words, the computational complexity of ETD is only $1.77\%$ of that of traditional encoding.
\end{example}

\section{Transformed Encoding of QC-LDPC Codes}

An LDPC code usually is defined by either its Tanner graph or its parity-check matrix. In this paper, we describe an ($en$,$K$) QC-LDPC code by its parity-check matrix, which can be simply denoted by a sparse $(n-k)\times n$ array of $e\times e$ circulants

    \begin{equation}
\begin{array}{lll}
             \bH &=&
\left[\begin{array}{cccc}
    \bA_{0,0}  &     \bA_{0,1} &  ... &     \bA_{0,n-1} \\
         \bA_{1,0} &   \bA_{1,1} &  ...&     \bA_{1,n-1}   \\
          \vdots& &\ddots & \vdots\\
                                       \bA_{n-k-1,0} &    \bA_{n-k-1,1}&  ...&  \bA_{n-k-1,n-1}
\end{array}\right],
                          \end{array}
\end{equation}
where $K\geq ek$. The authors of \cite{LCZLF06} proposed to compute the generator matrix $\bG$ with quasi-cyclic structure from the parity-check matrix based on Gaussian elimination. Based on the matrix transformation (\ref{eq:mt}), we can have the transformed generator matrix $\bG^{{\cal F},\pi}$ from the generator matrix $\bG$. Alternatively, we can have $\bG^{{\cal F},\pi}$ from the null space of the transformed parity-check matrix $\bH^{{\cal F},\pi}$ \cite{HLG_12}. In the sequel, we separately treat ETD of QC-LDPC codes with full rank $\bH$ and rank deficient $\bH$.

\subsection{Encoding of QC-LDPC Codes with Full Rank Parity-Check Matrices}

Suppose that $\cal C$ is an ($en$,$K$) QC-LDPC code defined by its full rank parity-check matrix $\bH=[\bA_{i,j}]$. Similar to (\ref{eq:mt}), we have the transformed parity-check matrix $\bH^{{\cal F}, \pi}=diag[ \bB_0, \bB_1, \ldots, \bB_{e-1} ]$.  Then we can construct the corresponding transformed generator matrix $\bG^{{\cal F}, \pi}=diag[ \bD_0, \bD_1, \ldots, \bD_{e-1} ]$ (with details in \cite{HLG_12}) from   $\bH^{{\cal F},pi}$ such that
$$
\bG^{{\cal F}, \pi} \bH^{{\cal F}, \pi \sf T} = diag[\bD_0\bB_0^{\sf T}, \bD_1\bB_1^{\sf T}, \ldots, \bD_{e-1}\bB_{e-1}^{\sf T} ]=0.
$$

To facilitate encoding, the diagonal matrices $\bD_i$'s should satisfy the conjugacy constraint and be systematic. We thus propose to compute the systematic matrix $\bD_i$ from the matrix $\hat{\bB}_i$, which is systemized from $\bB_i$ by Gaussian elimination.
In other words, if $\hat{\bB}_i=[ \bP_i  \; \vdots\; \bI ]$, then $\bD_i= [ \bI \; \vdots \; \bP_i^{\sf T}]$ such that $\bD_i \hat{\bB}_i^{\sf T} =0$. Furthermore, $\bD_{(2^\mu i)_e}= (\bD_{i})^{2^\mu}$.

Suppose that $\bH^{{\cal F}, \pi}$ has full rank, i.e., $K=ek$. It is clear that $\hat{\bB}_i=[ \bP_i  \; \vdots\; \bI ]$ for each $i=0,1,\ldots,e-1$ has full rank and all $\hat{\bB}_i$'s are of the same size. Then all the $\bD_i$'s have the same size. Thus, ETD of QC-LDPC codes in this case is the same as before in Algorithms 1 and 2.

\subsection{Encoding of QC-LDPC Codes with Rank Deficient Parity-Check Matrices}

The parity-check matrices of many QC-LDPC codes, especially algebraic QC-LDPC codes \cite{TLZLG06,LZTCLG07, HLG_10,Chen_Lan_Djurdjevic_Lin_04,SZLG09,JingyuKang2010,Huang_Diao_Lin_A-G,M.C.}, are rank deficient, i.e., $K\neq ek$.
In some cases, more than half rows of their parity-check matrices are redundant.
As a result, there exist rank deficient diagonal matrices in their transformed parity-check matrices.
Consequently, their diagonal matrices $\bB_i$'s have different ranks, with $\rho_i=rank(\bB_i)\leq n-k$, $i=0,1,\ldots,e-1$, such that
$$
\bB_i=\left[  \begin{array}{c}
\bP_i \; \vdots\; \bI_{\rho_i \times \rho_i}\\
{\bf 0}\\
\end{array}
\right]_{(n-k) \times n}.
$$
Thus, the diagonal matrices $\bD_i$'s on the transformed generator matrices have different sizes and ranks,
with $\sigma_i =rank(\bD_i)= n-\rho_i $.
 The length of the message vector of the QC-LDPC code is $K=\sum\limits ^{e-1}_{i=0} \sigma_i$ and the transformed generator matrix $\bG^{{\cal F}, \pi} = [\bD_0, \bD_1, \ldots, \bD_{e-1}]$
 can be represented by
$$
\bD_i= \left[\begin{array}{c}
\bI_{\sigma_i \times \sigma_i}   \; \vdots\;  \bP_i^{\sf T} \\
\end{array}
\right]_{\sigma_{i}\times \sigma_i}.
$$
Thus, ETD of QC-LDPC codes with rank deficient parity-check matrices is similar to Algorithms 1 and 2 as well.
The only difference lies in the first step because different $\bD_i$'s have different ranks and sizes.
Thus, $\sigma_i$ symbols or bits are multiplied by $\bD_i$ in the first step in the non-binary or binary case.

\section{Concluding Remarks}
In this paper, we have proposed a low-complexity encoding algorithm for QC codes in the transform domain.
Its computational complexity is much lower than traditional encoding, for both non-binary and binary QC codes.
To further simplify ETD of binary QC codes, a post-processing step is devised to guarantee that the transmitted codeword is binary. In addition, it has been shown that ETD is applicable to both QC-LDPC codes with full rank parity-check matrices and QC-LDPC codes with rank deficient parity-check matrices.

We have only considered cases with $e=2^r-1$ and $q=2^r$ or $q=2$ in this paper. The derivations for other cases are similar.
It is worth mentioning that, unlike traditional encoding, the transmitted codeword generated by our ETD algorithm is not systematic. Thus, after the codeword is corrected by the decoder of QC codes, Fourier transforms are required to recover the message vector.
The complexity of these Fourier transforms is much lower than that of traditional decoding algorithms for most QC codes. 
Moreover, non-systematic codewords by nature have better secrecy than systematic ones.

\appendix

\begin{theorem}
Let $\eta_i$ be the size of $i$-th conjugacy class. The nonlinear mapping from a binary vector $\bz$ to a vector $\hat{\bz}$ over GF($2^r$),
\begin{equation}
\hat{z}_{(2^\mu t)_e}  =  (\sum\limits^{\eta_i-1}_{l=0 }  \beta_l z_{ (2^l t)_e}   )^{2^ \mu},
\end{equation}
is one-to-one correspondence, where $\{ \beta_{l} \}$ is a basis of the subfield GF($2^{\eta_i}$). Furthermore, $\hat{\bz}$ satisfies the conjugacy constraint.
\end{theorem}
\begin{proof}
First, we prove that it is a one-to-one correspondence mapping. It can be proved by contradiction. Suppose that there exist two distinct vectors $\bz^{(1)}$ and $\bz^{(2)}$ which are mapped into the same vector $\hat{\bz}$, i.e.,
\begin{equation}
\begin{array}{lll}
(\sum\limits^{\eta_i-1}_{l=0 } \beta_l z_{ (2^l t)_e}^{(1)}  )^{2^ \mu}=(\sum\limits^{\eta_i-1}_{l=0 } \beta_l z_{ (2^l t)_e}^{(2)}  )^{2^ \mu}.\\
\end{array}
\end{equation}
Since $x^2+y^2=(x+y)^2$ holds in the extension fields of GF(2), we have
$$
\begin{array}{ccc}
\left(\sum\limits^{\eta_i-1}_{l=0 } \beta_l z_{ (2^l t)_e}^{(1)}-\sum\limits^{\eta_i-1}_{l=0 } \beta_l z_{ (2^l t)_e}^{(2)} \right)^{2^ \mu}&=&0,\\
\left(\sum\limits^{\eta_i-1}_{l=0 } \beta_l (z_{ (2^l t)_e}^{(1)} - z_{ (2^l t)_e}^{(2)})   \right)^{2^ \mu}                 &=&0,\\
\sum\limits^{\eta_i-1}_{l=0 } (z_{ (2^l t)_e}^{(1)} - z_{ (2^l t)_e}^{(2)}) \beta_l                                          &=&0.
\end{array}
$$
Since $\bz^{(1)}$ and $\bz^{(2)}$  are different, there exist nonzero coefficients $z_{ (2^l t)_e}^{(1)} - z_{ (2^l t)_e}^{(2)}$. It indicts that $\beta_l$'s are linearly dependent, which contradicts the assumptions that $\beta_l$'s are a basis. Similarly, it can be proved that any two different vectors $\hat{\bz}^{(1)}$ and $\hat{\bz}^{(2)}$ are mapped from two different vectors ${\bz}^{(1)}$ and ${\bz}^{(2)}$. As a result, the nonlinear mapping (25) is bijective.

Then we prove that $\hat{\bz}$ satisfies the conjugacy constraint as follows.

 For $0\le \mu < \eta_i-1$, where $\eta_i>1$, clearly, $\hat{z}_{(2^{\mu+1}t)_e}=\hat{z}_{{(2^\mu t)}_e}^2$. For $\mu=\eta_i-1$, we obtain that
$$
\begin{array}{ccc}
\hat{z}_{(2^{\eta_i-1} t)_e}^2&=&\left((\sum\limits^{\eta_i-1}_{l=0 }  \beta_l z_{ (2^l t)_e}   )^{2^ {\eta_i-1}}\right)^2\\
&=&\sum\limits^{\eta_i-1}_{l=0 }  \beta_l^{2^{\eta_i}} z_{ (2^l t)_e}^{2^{\eta_i}}.
\end{array}
$$
Since  $\beta_l$ is over GF($2^{\eta_i}$), we obtain $\beta_l^{2^{\eta_i}}=\beta_l$. Since $\bz$ is a binary vector, $z_{ (2^l t)_e}^{2^{\eta_i}}=z_{ (2^l t)_e}$. Thus
$$
\hat{z}_{(t)_e}=\hat{z}_{{(2^\mu t)}_e}^2.\\
$$
Hence we can obtain that
\begin{equation}
\hat{z}_{(2^{\mu+1}t)_e}=\hat{z}_{{(2^\mu t)}_e}^2,\\
\end{equation}
for $0 \le \mu<\eta_i$.

Since (27) holds for all $i$, it follows that $\hat z_{(2t)_e}=\hat z_t^2$ for $0\le t<e$. This is equivalent to the condition that $\hat{\bz}$ satisfies the conjugacy constraint.\end{proof}


\begin{thebibliography}{11}

\bibitem{Townsend_67}
R. L. Townsend and E. J. Weldon, Jr., ``Self-orthogonal quasi-cyclic codes,'' \emph{IEEE Trans. Inform. Theory}, vol. IT-13, no. 2, pp. 183-195, Apr. 1967.

\bibitem{Kasami74}
T. Kasami,
``A Gilbert-Varshamov bound for quasi-cycle codes of rate 1/2,''
{\it IEEE Trans. Inform. Theory}, vol. IT-20, no. 5, p. 679, Sep. 1974.

\bibitem{LC04} S. Lin and D. J. Costello, Jr., \emph{Error Control Coding: Fundamentals and
    Applications}, 2nd edition. Upper Saddle River, NJ: Prentice Hall, 2004.

\bibitem{Bon08}
N. Bonello, C. Sheng Chen, and L. Hanzo, ``Construction of regular quasi-cyclic protograph LDPC codes based on Vandermonde matrices,'' \emph{IEEE Trans. Vehicular Technology}, vol. 57,  no. 4, pp. 2583--2588, Jul. 2008.


\bibitem{TLZLG06} Y. Y. Tai, L. Lan, L. Zheng, S. Lin and K. Abdel-Ghaffar,
`` Algebraic construction of quasi-cyclic LDPC codes for the AWGN and erasure channels,''
\emph{IEEE Trans. Commun.}, vol 54, no. 7, pp. 1765--1774, Oct. 2006.
 
\bibitem{N.B08}
 J. L. Fan, ``Array codes as low-density parity-check codes,'' in \emph{proc. Int. Symp. Turbo Codes,}  Brest, France, Sep. 2-7, 2000, pp.545-546.
 
 \bibitem{MyungYang05}
S. Myung and K. Yang,
`` A combining method of quasi-cyclic LDPC codes by the Chinese remainder theorem,''
\emph{IEEE Commun. Lett.}, vol. 9, no. 9, pp. 823--825, Sep. 2005.

\bibitem{LalFitz01}
K. Lally and P. Fitzpatrick, ``Algebraic structure of quasicyclic codes,''
\emph{ Disc. Appl. Math.}, vol. 111, pp. 157--175, 2001.

\bibitem{JonWell03}
S. J. Johnson and S. R. Weller,
``A family of irregular LDPC codes with low encoding complexity,''
\emph{IEEE Commun. Lett.}, vol. 7, no. 2, pp. 79--81, Feb. 2003.

\bibitem{YR04}
M. Yang and W. E. Ryan,
``Performance of efficiently encodable low-density parity-check
codes in noise bursts on the EPR4 channel,''
{\it IEEE Trans. Magn.}, vol. 40, no. 2, part 1, pp. 507--512. Mar. 2004.

\bibitem{LZTCLG07} L. Lan, L. Zeng, Y. Y. Tai, L. Chen, S. Lin, and K. Abdel-Ghaffar, ``Construction of
  quasi-cyclic LDPC codes for AWGN and binary erasure channels: A finite field approach,''
  \emph{IEEE Trans. Inform. Theory}, vol. 53, no. 7, pp. 2429--2458, Jul. 2007.


\bibitem{LCZLF06}  Z. Li, L. Chen, L. Zeng, S. Lin and W. Fong, ``Efficient encoding of quasi-cyclic low-density parity-check codes,''
\emph{IEEE Trans. Commun.}, vol. 54, no. 1, pp.  71--81, 2006.

\bibitem{K07} N. Kamiya, ``High-rate quasi-cyclic low-density parity-check codes derived from finite affine planes,'' \emph{IEEE
          Trans. Inform. Theory}, vol. 53, no. 4, pp. 1444--1459, Apr. 2007.

\bibitem{HLG_10}
 Q. Diao, Q. Huang, S. Lin, and K. Abdel-Ghaffar,
``A transform approach for computing the ranks of parity-check matrices of quasi-cyclic LDPC codes,'' \emph{Proc. 2011 IEEE Int. Symp. Inform. Theory}, SaintPetersburg, Russia, pp. 366-379, July 31-Aug. 5, 2011.

\bibitem{HLG_12}
 Q. Diao, Q. Huang, S. Lin, and K. Abdel-Ghaffar,
``A matrix theoretic approach for analyzing quasi-cyclic low-density parity-check codes,''
\emph{IEEE Trans. Inform. Theory}, vol. 58, no. 6, pp. 4030--4048, June. 2012.


\bibitem{HLW_12}
Q. Huang, K. Liu, and Z. Wang, ``Low-density arrays of circulant matrices: Rank and row-redundancy, and QC-LDPC
codes,'' in \emph{2012 IEEE International Symposium on Information Theory}, pp. 3073-3077, 2012.

\bibitem{B83}  R. E. Blahut,  \emph{Theory and Practice of Error Control Codes}. Reading, MA: Addison-Wesley, 1983.


\bibitem{Chen_Lan_Djurdjevic_Lin_04}
L. Chen, L. Lan, I. Djurdjevic, and S. Lin, ``An algebraic method for construction quasi-cyclic LDPC codes,''
Proc. \textit{Int. Symp. Inform. Theory and Its Applications}, Parma, Italy, Oct. 10--13, 2004, pp. 535--539.

\bibitem{SZLG09} S. Song, B. Zhou, S. Lin, and K. Abdel-Ghaffar,
``A unified approach to the construction of binary and nonbinary quasi-cyclic LDPC codes based on finite fields,''
\emph{IEEE Trans. Commun.}, vol. 57, no. 1, pp. 84--93, Jan. 2009.

\bibitem{JingyuKang2010}
J. Y. Kang, Q. Huang, L. Zhang, B. Zhou, and S. Lin, ``Quasi-Cyclic
LDPC Codes: An Algebraic Construction,'' {\it IEEE
Trans. Commun.}, vol. 58, no. 5, pp. 1383--1396, May. 2010.

\bibitem{Huang_Diao_Lin_A-G}
Q. Huang, Q. Diao, S. Lin, and K. Abdel-Ghaffar,
``Cyclic and quasi-cyclic LDPC codes on constrained parity-check matrices and their trapping sets,''
\emph{IEEE Trans. Inform. Theory}, vol. 58, no. 5, pp. 2648--2671, May. 2012.

\bibitem{NA}
\emph{Low Density Parity Check Codes for Use in Near-earth and Deep Space Appliations.} Recommendation for Space Data System Standards, CCSDS 131.1-O-2. Blue Book. Washington, D.C.: CCSDS, Sep. 2007.


\bibitem{Ga62}
R. G. Gallager,
``Low density parity check codes,''
{\it IRE Trans. Inform. Theory}, vol. IT-8, no. 1, pp. 21--28, Jan. 1962.

\bibitem{Tan88}
R. M. Tanner,
``A transform theory for a class of group-invariant codes,''
\emph{IEEE Trans. Inform. Theory}, vol. 34, no. 4, pp. 752--775, Jul. 1988.

\bibitem{RO06}
R. M. Roth,
\emph{Introduction to Coding Theory}. Cambridge, U.K.: Cambridge Univ. Press, 2006.

\bibitem{KT1990}
R. Kao, F. J. Taylor, ``A fast Galois-field  transform algorithm using normal bases,'' \emph{IEEE Twenty-Fourth Asilomar Conference on Signals, Systems and Computers,}  Nov. 5-7, 1990.

\bibitem{M.C.}
 M. C. Davey and  D. J. C. Mackay, ``Low density parity check codes over GF(q),'' \emph{IEEE Information Theory Workshop}, pp.70-71, Jun. 1998.
 

  
\end{thebibliography}
\end{document}